\documentclass[aps,pra,reprint,superscriptaddress,nofootinbib,floatfix]{revtex4-2}

\usepackage{amsmath,amssymb,mathtools}
\usepackage{amsthm}
\usepackage{pgfplots}
\usepackage[colorlinks=true,linkcolor=blue,citecolor=blue,urlcolor=blue]{hyperref}
\pgfplotsset{compat=1.18}

\newtheorem{proposition}{Proposition}

\newcommand{\conv}{\operatorname{conv}}
\newcommand{\STAB}{\operatorname{STAB}}

\begin{document}

\title{Operational Shadows of Hilbert-Space Probabilities}

\author{Karl Svozil}
\affiliation{Institute for Theoretical Physics, TU Wien,
Wiedner Hauptstra{\ss}e 8-10/136, 1040 Vienna, Austria}
\email{karl.svozil@tuwien.ac.at}

\date{\today}

\begin{abstract}
At one frozen setting, the probabilities observed in a sharp quantum context are indistinguishable, as detector-click statistics, from ordinary probabilities on the atoms of a classical partition. But an actual analyzer usually comes with a calibrated knob: a tangible handle on the apparatus. If the measurement configuration is co-varied continuously through this physical parameter, the operational object is no longer one point of a simplex but a response curve. Classical linear responses, Malus-type Hilbert-space responses, softmax links, non-homomorphic parameter transcriptions, and discontinuous threshold limits are different maps from settings to probabilities. Continuity, calibration, and preservation of the physical composition law are then part of the experimental meaning of the knob. Such comparisons distinguish specified, calibrated response models; by themselves they do not constitute a classical-versus-quantum impossibility theorem. The static operational coincidence can also persist for two intertwined contexts: if the common outcomes receive the same probabilities, the remaining masses can always be coupled by a classical joint distribution. Genuine multi-context nonclassicality begins when a family of local shadows cannot be glued into one nonnegative global distribution or one simplex factorization. Farkas' lemma gives the exact alternative: either the classical extension exists, or a separating linear inequality certifies its impossibility.
\end{abstract}

\maketitle

\section{Introduction}

A hand held before a lamp casts a flat shadow on a wall.  Many different
hands, postures, and objects can cast the same silhouette.  If the wall is the
only thing one is allowed to inspect, the shadow is all one has; but it should
not be mistaken for the object that produced it.  This is the old warning of
Plato's cave: the prisoners take the passing shadows for reality.  The image
contains a second warning that is just as relevant here.  The one who turns
around and returns with an account of the objects and the light behind the
shadows is not automatically believed; he may appear blinded or confused, and
the cave's inhabitants may punish the attempt to free them from their own
evidence~\cite[Book VII, 514a--517a]{plato-republic}.  Thus the metaphor cuts
both ways.  A shadow should not be reified as the thing itself, but neither is
the hidden machinery behind it given by the shadow alone.  To speak of an
ontology casting the shadow is already to add a reconstruction, a model of
what lies off the wall.

Something similar happens in a single sharp measurement.  A system is sent
into a detector bank, one detector clicks, and after many repetitions one has
a list of relative frequencies.  This list is a point of a probability
simplex.  Looking only at that list, one cannot tell whether it was generated
by a classical partition, by squared Hilbert-space amplitudes, or by a more
exotic admissible weight rule.  The generating geometries are different, but
a single frozen setting can flatten them into the same operational shadow.

The missing character is often the knob: a tangible handle on the apparatus.
A polarizer rotates, a Stern--Gerlach magnet tilts, a phase is advanced, or,
in a more abstract response model, a temperature or score is dialed.  In the
rotational examples, turning the knob is itself part of the experiment: equal
increments can be calibrated, small changes can be made continuously, and
successive transformations compose.  The knob is therefore not merely a
coordinate on an interval; it carries a physical action.  Once this variation
is retained, the shadow is no longer a single point but a trajectory through
the simplex.  A specified classical response model may hit any one point of a
Hilbert-space response curve, but it need not reproduce the whole curve with
the same symmetry-preserving knob.  This is the intuition behind the
``parameter cheats'': they work by deforming the scale on which the knob is
read, thereby changing the physical action represented by equal turns of the
apparatus.

This article separates four levels.  One static context gives an operational
shadow.  A calibrated sweep compares response laws within a specified model
class.  Two compatible context distributions always admit a classical
coupling.  Only a sufficiently constrained web of contexts can fail to extend
to one global nonnegative distribution, in which case Farkas' lemma yields a
separating inequality.  The response-law discussion is methodological: a
single calibrated curve does not, without additional covariance, locality, or
ontic-response assumptions, certify Hilbert-space probability or exclude
arbitrary classical stochastic models.

A continuously swept apparatus can indeed be represented formally as a
continuum of mutually incompatible quantum contexts: each fixed setting
determines a different maximal projective measurement.  That observation is
not, however, the main use of the knob in this article.  The sweep is treated
first as an experimentally calibrated response measurement---a map from a
physical control variable to click probabilities---rather than as a
quantum-logical pasting problem over a continuum of contexts.  The group
action and continuity belong to the kinematics of the apparatus; they do not
by themselves supply a nonclassicality witness or an
orthomodular-algebraic obstruction.  Such an obstruction arises only when one
asks whether the probabilities from a specified family of contexts can be
glued into one global nonnegative model.

\section{One static context}

Let \(C=\{P_1,\ldots,P_d\}\) be a maximal rank-one projective
measurement on a \(d\)-dimensional Hilbert space.  In a run of the experiment
only one detector clicks, and the operational data associated with this
single context are the frequencies
\begin{equation}
        p_i\ge 0,\qquad \sum_{i=1}^d p_i=1 .
\end{equation}
They are the same kind of object as a probability distribution on the atoms
of a classical partition.  The phrase ``same operational shadow''%
\footnote{This use of ``shadow'' is the Plato-cave metaphor used here, not
the randomized measurement technique known in quantum information as
``classical shadows''.}
will mean this: the observable click statistics in the context are identical,
even if the underlying state spaces used to generate them are different.

The difference between the two theories is already present in the generating
maps.  A classical state on a \(d\)-atom partition is a point of the simplex
\(\Delta_{d-1}\), with pure states at the vertices and mixtures represented
affinely.  A pure quantum state is a ray in complex projective space,
\([\psi]\in\mathbb{CP}^{d-1}\), and Born's rule gives
\begin{equation}
        s_C([\psi]) =
        \bigl(\langle\psi|P_1|\psi\rangle,\ldots,
        \langle\psi|P_d|\psi\rangle\bigr).
        \label{eq:shadow-map}
\end{equation}
In an eigenbasis \(P_i=|e_i\rangle\langle e_i|\), this is the squared
projection map \(p_i=|\langle e_i|\psi\rangle|^2\).  Thus the same probability
simplex is reached through a different geometry: affine weights in the
classical case, squared Hilbert-space amplitudes in the quantum case.

\begin{proposition}
For one maximal rank-one context \(C\), the pure-state shadow map
\(s_C:\mathbb{CP}^{d-1}\to\Delta_{d-1}\) is onto.
\end{proposition}

\begin{proof}
Given \(p=(p_1,\ldots,p_d)\in\Delta_{d-1}\), choose arbitrary phases
\(\phi_i\) and set
\begin{equation}
        |\psi\rangle
        =
        \sum_{i=1}^d \sqrt{p_i}\,e^{i\phi_i}|e_i\rangle .
\end{equation}
Then \(\langle\psi|P_i|\psi\rangle=p_i\).  The phases, modulo a global phase,
are invisible to this context.  Hence a single context sees the full ordinary
probability simplex and nothing beyond it.
\end{proof}

This proposition is deliberately modest.  It does not say that the classical
and quantum state spaces are isomorphic.  It says only that after forgetting
everything except one detector bank's click frequencies, both theories cast
the same operational shadow.  The phrase ``forgetting'' is essential: it
freezes the apparatus setting and discards the way in which the probabilities
change when the apparatus is moved.

\section{Calibrated response laws}

In the laboratory the setting is rarely just an abstract label.  A polarizer
is rotated, a Stern--Gerlach magnet is tilted, a phase is advanced, or a score
parameter is changed.  The detector bank may still be the same operational
device, but the measurement configuration is co-varied with a physical
parameter.  Once this knob is retained, the observed data are not merely
\(p\in\Delta_{d-1}\), but the curve swept out by \(p\) as the knob is turned.
If \(K\) denotes the setting space of the knob, the relevant object is the
response function
\begin{equation}
        R_C:K\longrightarrow \Delta_{d-1},\qquad
        k\longmapsto (p_1(k),\ldots,p_d(k)).
        \label{eq:response-function}
\end{equation}
A classical partition can hit any one point of this curve.  That is not the
same as reproducing the curve with the same independently calibrated physical
action.

A calibrated knob is not just an arbitrary coordinate on an interval.  More
generally, its setting space \(K\) carries the action of a physically
implemented transformation group \(G\), and is often a homogeneous space
\(G/H\).  An admissible change of coordinates should preserve this action, or
at least intertwine it with the corresponding physical action on the
apparatus.  Continuity is part of this structure, but it is not the whole
structure: a continuous deformation of the coordinate can still fail to
preserve the action.

For a one-parameter planar rotation, the composable knob variable should be a
signed angle
\(\varphi\in\mathbb{R}/2\pi\mathbb{Z}\), or the corresponding quotient
appropriate to unoriented axes.  The scalar \(\theta\in[0,\pi]\) used below
is an unsigned relative separation derived from that action.  It is suitable
as the argument of an isotropic response curve, but it does not itself obey a
global addition law.  For general analyzer orientations the acting group is
\(SO(3)\), and the scalar separation likewise should not be identified with a
group element.  A harmless reparameterization must therefore intertwine the
physical group action, not merely relabel the scalar argument of a plotted
curve.  Such calibration can be checked independently of the output
probabilities by composing transformations, verifying the appropriate
period, or using a phase reference.  Equal increments, composed rotations,
and continuous sweeps are operationally implementable and comparable.

One may object that, once the knob is turned, one has already left the
original context and entered a many-context situation.  Formally this can be
a useful description: each fixed angle may be represented by its own maximal
context.  But that is not the emphasis here.  A knob sweep is not primarily a
finite list of unrelated configurations to be glued by a marginal-extension
test.  It is a functional response to an argument, such as
\(\theta\mapsto p(\theta)\), where the argument has a physical scale,
continuity, and composition or symmetry law.  The question is therefore not
only which static contexts are present, but how the probabilities vary when
the apparatus is continuously moved.

Within a specified model class, this is already enough to distinguish
response laws along a single calibrated sweep.  For example, a deterministic
shared-direction model with a uniformly distributed hidden orientation and a
threshold readout gives a linear equal-outcome response, whereas a spin-\(1/2\)
singlet gives the curved Malus-type response.  This comparison does not
exclude all classical stochastic response models for one curve; locality and
global-consistency restrictions enter only when several settings must be
represented together.  For these binary examples write
\(R(\theta)=P^{=}(\theta)\).  On \(0\le\theta\le\pi\) one may write
\begin{equation}
\begin{aligned}
        P^{=}_{\rm cl}(\theta)&=\frac{\theta}{\pi},\\
        P^{=}_{\rm qm}(\theta)&=\sin^2\frac{\theta}{2},\\
        P^{=}_{\rm s}(\theta)&=
        H\!\left(\frac{2\theta}{\pi}-1\right),
\end{aligned}
\label{eq:three-response-curves}
\end{equation}
where \(H(0)=1/2\) fixes the convention at the threshold.  The last line is a
discontinuous toy response, sometimes used as a stronger-than-quantum limit
in a multi-setting construction~\cite{svozil-krenn}; the curve by itself is
not a no-signalling multi-setting model.  At a single value of \(\theta\),
each line gives only an ordinary two-outcome probability distribution.  As
\(\theta\) is co-varied, however, the variational response is completely
different: linear, smoothly curved, or step-like.

The 2001 ``parameter cheat'' construction makes this point particularly
transparent.  If one introduces a new parameter \(\delta\) by
\(\theta=\pi\sin^2(\delta/2)\), then the classical expression
\(P^{=}_{\rm cl}(\theta)=\theta/\pi\) can be made to look quantum:
\[
        P^{=}_{\rm cl}(\theta(\delta))
        =
        \sin^2\frac{\delta}{2}.
\]
Conversely, the quantum curve can be made to look linear by displaying it
against the parameter
\(\eta=\pi\sin^2(\theta/2)\).  Pointwise, the trick works.  But the price is
that the new parameter is no longer the same physical knob.  The
transcription is not a homomorphism, or more generally an equivariant map, for
the calibrated rotation action~\cite{svozil-2001-cesena}.  It therefore
changes the calibration through which the apparatus is varied rather than
merely relabeling probabilities.

\subsection{Softmax links and inverse knobs}

Softmax belongs to the same discussion, although its natural knob is often a
score, utility, or inverse-temperature parameter rather than a literal angle.
It is included here only as an abstract response law into a simplex, not as a
physical analogue of Born's quadratic amplitude rule.  Given scores
\(u_C(a)\) and a strictly positive, usually monotone, link \(g\), it forms
\begin{equation}
        P_{g,C}(a)
        =
        \frac{g(u_C(a))}
        {\sum_{b\in C}g(u_C(b))}.
        \label{eq:local-softmax}
\end{equation}
For \(g_\beta(u)=\exp(\beta u)\), the parameter \(\beta\) interpolates between
an almost uniform response at small \(\beta\) and a winner-take-all threshold
as \(\beta\to\infty\).  This is useful as a link-function deformation, not as
a claim that the exponential family has the same physical origin as Born's
rule.

Changing \(g\), or changing the score scale, changes the response sensitivity
of the model.  For a single context this again guarantees only local
normalization; it does not by itself specify which score scale is physically
proper.  If the same atom occurs in several contexts, admissibility adds a
separate requirement: the response must be single-valued on that atom.  Thus
a softmax fit, a Malus curve, a linear classical curve, and a Heaviside step
are not competing physical explanations.  They are different response
functions tied to different assumptions about what is being varied and which
structure that variation is required to preserve.

For a single monotone response curve this observation can be made completely
explicit.  Let \(R_A(\theta)\) and \(R_B(\theta)\) be two response functions,
written here as equal-outcome probabilities.  Equivalently, one may use the
correlation \(E(\theta)=2R(\theta)-1\).  If \(R_A\) is invertible on the range
of \(R_B\), then
\begin{equation}
        T_{A\leftarrow B}(\theta)
        =
        R_A^{-1}(R_B(\theta))
        \label{eq:inverse-knob}
\end{equation}
is the transcribed knob: feeding \(T_{A\leftarrow B}(\theta)\) into response
law \(A\) makes it reproduce response law \(B\) as a function of the
displayed parameter \(\theta\).  The injectivity assumption is substantive.
Without a single-valued inverse one obtains at best a set-valued
transcription or a limiting prescription, as in a threshold response.

For the classical linear response
\(R_{\rm cl}(\theta)=\theta/\pi\) and the quantum singlet response
\(R_{\rm qm}(\theta)=\sin^2(\theta/2)\), this gives
\begin{align}
        T_{{\rm cl}\leftarrow{\rm qm}}(\theta)
        &=\pi\sin^2\frac{\theta}{2},&
        T_{{\rm qm}\leftarrow{\rm cl}}(\theta)
        &=2\arcsin\sqrt{\frac{\theta}{\pi}}.
        \label{eq:inverse-classical-quantum}
\end{align}
These are precisely the parameter cheats: the first makes a classical linear
response look quantum; the second makes the quantum response look classical.

For a centered binary softmax response with parameter \(\beta > 0\), define
\begin{equation}
        P^{=}_{\sigma,\beta}(\theta)
        \equiv R_{\sigma,\beta}(\theta)
        =
        \left\{1+\exp[-\beta(2\theta/\pi-1)]\right\}^{-1}.
        \label{eq:binary-softmax-angle}
\end{equation}
This is the softmax counterpart of the response laws in
Eq.~\eqref{eq:three-response-curves}.  As \(\beta\downarrow0\) it approaches
the uniform (constant-in-$\theta$) binary response \(P^{=}(\theta)=1/2\); as \(\beta\to\infty\) it
approaches the Heaviside threshold response (with convention \(H(0)=1/2\)).

The corresponding inverse knobs are
\begin{align}
        T_{{\rm cl}\leftarrow\sigma}(\theta)
        &=
        \pi R_{\sigma,\beta}(\theta),
        \label{eq:inverse-classical-softmax}\\
        T_{\sigma\leftarrow{\rm cl}}(\theta)
        &=
        \frac{\pi}{2}
        \left[
        1+\frac{1}{\beta}
        \log\frac{\theta/\pi}{1-\theta/\pi}
        \right],
        \label{eq:inverse-softmax-classical}
\end{align}
where the second formula is finite only for \(0<\theta/\pi<1\).  For an input
knob restricted to \([0,\pi]\), the inverse
\(T_{\sigma\leftarrow{\rm cl}}(\theta)\) lies in that physical interval
precisely when
\begin{equation}
        \frac{1}{1+e^\beta}
        \le
        \frac{\theta}{\pi}
        \le
        \frac{1}{1+e^{-\beta}}.
        \label{eq:softmax-inverse-domain}
\end{equation}
Outside this range the inverse is still a formal transcription, but not an
admissible setting of the bounded softmax knob.  Similarly,
\[
        T_{{\rm qm}\leftarrow\sigma}(\theta)
        =
        2\arcsin\sqrt{R_{\sigma,\beta}(\theta)},
\]
and \(T_{\sigma\leftarrow{\rm qm}}\) is obtained by inserting
\(\sin^2(\theta/2)\) into the logit expression in
Eq.~\eqref{eq:inverse-softmax-classical}.  The Heaviside response
\(R_{\rm H}(\theta)=H(2\theta/\pi-1)\) can be transcribed into the classical
knob as
\(T_{{\rm cl}\leftarrow{\rm H}}(\theta)=\pi R_{\rm H}(\theta)\), but it has
no ordinary inverse: an interval of angles is collapsed to one of two values.
Its inverse is set-valued, or else appears only as a limiting threshold.

\begin{figure}[htbp]
\centering
\begin{tikzpicture}
\begin{axis}[
    width=\columnwidth,
    height=0.52\columnwidth,
    xmin=0,xmax=1,
    ymin=0,ymax=1,
    xlabel={displayed angle \(\theta/\pi\)},
    ylabel={classical knob \(T_{{\rm cl}\leftarrow B}(\theta)/\pi\)},
    legend style={font=\scriptsize,at={(0.02,0.98)},anchor=north west},
    ticklabel style={font=\scriptsize},
    label style={font=\scriptsize},
    samples=200
]
\addplot[black,dashed,domain=0:1] {x};
\addlegendentry{classical}
\addplot[blue,thick,domain=0:1] {sin(90*x)^2};
\addlegendentry{quantum}
\addplot[red,thick,domain=0:1] {1/(1+exp(-6*(2*x-1)))};
\addlegendentry{$\;$softmax \(\beta=6\)}
\addplot[orange,thick] coordinates {(0,0) (0.5,0) (0.5,1) (1,1)};
\addlegendentry{Heaviside}
\end{axis}
\end{tikzpicture}
\caption{Classical inverse knobs for several target response laws.  Since
\(R_{\rm cl}(\tau)=\tau/\pi\), the transcribed classical knob is simply
\(T_{{\rm cl}\leftarrow B}(\theta)/\pi=R_B(\theta)\).  The curves show how a
linear classical response must be read on a generally
non-symmetry-preserving scale in order to mimic a quantum, softmax, or
Heaviside response as a function of the displayed angle.}
\label{fig:inverse-knobs}
\end{figure}
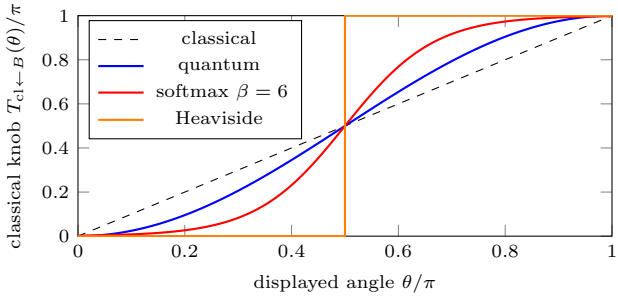

This single-curve trick should not be confused with a solution of a Bell or
Boole problem.  In a Bell--Clauser--Horne--Shimony--Holt (CHSH) experiment one
does not merely fit one function \(E(\theta)\).  One must assign probabilities
to several settings at once, such as \(a,b,a',b'\), with one common
underlying distribution over predetermined answers.  A separate inverse knob
can make each pairwise correlation look right, but those pairwise
deformations need not come from one action-equivariant calibration.  For a
signed planar parameter \(\varphi\), the transcribed map generally fails to
satisfy
\[
        T(\varphi_1+\varphi_2)
        =
        T(\varphi_1)+T(\varphi_2);
\]
for general analyzer orientations, the corresponding failure is
nonintertwining of the \(SO(3)\) action.  Thus the deformed pairwise settings
need not be compatible with one geometry of four settings.  The
marginal-extension problem then reappears, and Farkas' lemma supplies the
relevant obstruction.

Elastic-band examples make this warning concrete.  Aerts' spin machine and
its extended-Bloch descendants use a hidden break point, or more generally a
hidden measurement interaction, to generate Born-like responses from an
ordinary macroscopic mechanism~\cite{aerts-69,Aerts_2014,Aerts_2016}.  The
rubber-band analysis of context communication stresses the complementary
limitation: recovering a cosine-shaped pair correlation in one adapted
arrangement is not yet the same as a Bell-local simulation of an entire CHSH
family~\cite{svozil-2022-epr}.  If the band or shared resource has to be
realigned when the selected pair
\((a,b),(a,b'),(a',b),(a',b')\) changes, then the construction has acquired
setting dependence or context communication.  The shadow can be made to look
quantum after breakfast, lunch, teatime, and supper, but the four shadows no
longer arise from one setting-independent joint valuation.

In Sven Aerts' Popescu--Rohrlich-box rubber-band construction, the same point
appears in a more extreme form: the resource is one extended object whose
breaking co-creates the two outcomes~\cite{AertsSven}.  These models are
relevant here not as refutations of Boole--Bell, but as physical parables of
where the response-function program stops.  A single curve can be
reparameterized, whereas a Bell web requires one common nonadaptive geometry
or else a setting-dependent, nonlocal, or nonseparable resource.

\section{Two contexts and global extensions}
\label{sec:two-contexts}

The same logic explains why two contexts, even when they share outcomes, are
not yet enough to force a probability-level obstruction.  Let
\begin{equation}
        A=S\cup X,\qquad B=S\cup Y
\end{equation}
be two Boolean blocks, with shared atoms \(S\), atoms \(X\) belonging only to
\(A\), and atoms \(Y\) belonging only to \(B\).  Suppose context-wise
probabilities \(p^A\) and \(p^B\) are normalized, nonnegative, and agree on
the overlap:
\begin{equation}
        p^A_s=p^B_s=:r_s\qquad (s\in S).
        \label{eq:overlap}
\end{equation}
Let \(m=1-\sum_{s\in S}r_s\).  Then
\begin{equation}
        \sum_{x\in X}p^A_x=m=\sum_{y\in Y}p^B_y .
\end{equation}
If \(m>0\), choose, for example,
\begin{equation}
        w_{xy}=\frac{p^A_x p^B_y}{m};
        \label{eq:two-context-coupling}
\end{equation}
if \(m=0\), set all \(w_{xy}=0\).  Now define a classical sample space with
one point \(\lambda_s\) of weight \(r_s\) for each shared atom \(s\), and one
point \(\lambda_{xy}\) of weight \(w_{xy}\) for each pair
\(x\in X,y\in Y\).  At \(\lambda_s\), the shared atom \(s\) is the true answer
in both contexts.  At \(\lambda_{xy}\), the answer to \(A\) is \(x\) and the
answer to \(B\) is \(y\).  Reading the \(A\)-coordinate gives \(p^A\), and
reading the \(B\)-coordinate gives \(p^B\).

\begin{proposition}
Every pair of normalized context distributions satisfying
Eq.~\eqref{eq:overlap} has a classical joint model over deterministic answers
for the two blocks.
\end{proposition}

This is just the elementary coupling theorem for two marginal distributions,
with the shared outcomes treated as already coupled.  Equivalently, after the
common mass has been fixed, each residual cell obeys the finite
Fr\'echet--Hoeffding bounds
\begin{equation}
        \max(0,p^A_x+p^B_y-m)
        \le w_{xy}\le
        \min(p^A_x,p^B_y),
        \label{eq:frechet-hoeffding}
\end{equation}
and Eq.~\eqref{eq:two-context-coupling} is just one point in this coupling
polytope.  Thus two-context intertwining has no nontrivial marginal
obstruction at the level of static weights.  The response functions realizing
those weights may still be very different: continuous or discontinuous,
affine or Malus-type, softmax-shaped or boundary-valued.  The obstruction
requires a sufficiently constrained web of local shadows, typically a cycle
or a Kochen--Specker-type pasting, so that all pairwise or context-wise
agreements cannot be realized by one global nonnegative distribution.

\subsection{Farkas alternative}

The extension problem can be written in the form
\begin{equation}
        Mq=b,\qquad q\ge 0,
        \label{eq:marginal-lp}
\end{equation}
where \(q\) is a nonnegative distribution over deterministic global
assignments and \(b\) collects the observed context probabilities, including
the relevant normalization constraints.  Farkas' lemma says that exactly one
of the following alternatives holds:
\begin{align}
&\exists q\ge 0 \text{ such that } Mq=b,                         \\
&\exists y \text{ such that } y^TM\ge 0 \text{ and } y^Tb<0 .
        \label{eq:farkas}
\end{align}
The second line is a separating inequality.  Garg and Mermin used precisely
this logic to formulate the question of whether specified pair distributions
can arise as marginals of compatible higher-order
distributions~\cite{Garg1984}.  Fine's theorem for Bell inequalities is a
celebrated special case of the same joint-distribution
viewpoint~\cite{Fine-82}; Pitowsky's correlation polytopes place the
construction in convex-geometric form~\cite{pitowsky-86}.

Equation~\eqref{eq:marginal-lp} is also the finite-dimensional skeleton of
simplex-embeddability tests.  A finite prepare-and-measure fragment is
classical when its probability data admit a factorization through a simplex.
Once the relevant extremal classical response assignments and linear
operational equivalences have been enumerated, the unknown convex weights
obey a feasibility system of the form~\eqref{eq:marginal-lp}.  Infeasibility
then has a Farkas dual certificate.  Recent generalized-probabilistic simplex
criteria and their linear-programming implementations make this operational
version explicit~\cite{schmid-2021-simplex,selby-2024-lp}.

\section{The pentagon: classical, quantum, and admissible bodies}

For a fixed exclusivity structure there are three natural nested convex sets:
a classical body, a Hilbert-space body, and a larger admissible-weight body.
The classical body is generated by deterministic two-valued states:
\begin{equation}
        \mathcal P_{\rm cl}
        =
        \conv\{\text{admissible }0/1\text{ assignments}\}.
\end{equation}
For graphs this is the stable-set polytope, denoted \(\STAB(G)\), possibly
with additional normalization equations for complete contexts.

One standard realization of the corresponding quantum body is obtained from
squared projections of a state onto an orthogonal representation:
\begin{equation}
        p_v=|\langle\psi|v\rangle|^2,
\end{equation}
or by the corresponding density-operator formula.  In graph language this is
related to Lov\'asz-type theta bodies
\cite{Lovasz1975,GroetschelLovaszSchrijver1986,cabello-severini-winter-2014}.
The still larger admissible-weight body, often described by no-disturbance
constraints in contextuality scenarios, is obtained by retaining only
nonnegativity, context-wise normalization, additivity on exclusive events,
and equality of the values assigned to the same atom in different contexts.

For one clique, these three bodies have the same operational shadow: the
simplex.  For two cliques pasted along a common clique, the explicit coupling
establishes only that there is no classical extension obstruction for the
compatible context marginals.  It does not imply that every such pair is
realizable by a fixed pair of quantum projective contexts.

Cycles change the situation.  The pentagon is the minimal cyclic illustration.
Let
\[
        C_i=\{a_i,a_{i+1},x_i\},\qquad i=1,\ldots,5,
\]
with indices modulo \(5\).  The cyclic atoms \(a_i\) are the intertwining
observables, and the atoms \(x_i\) are context-specific.  The assignment
\begin{equation}
        p(a_i)=\frac12,\qquad p(x_i)=0
        \quad (i=1,\ldots,5)
        \label{eq:pentagon-half-weight}
\end{equation}
obeys completeness in every context,
\[
        p(a_i)+p(a_{i+1})+p(x_i)=1,
\]
is additive on mutually exclusive events inside each Boolean block, and is
atom-consistent: the same intertwining atom has the same value wherever it
occurs.  This is deliberately weaker than noncontextual hidden-variable
embeddability in a classical simplex, and it is not Spekkens-style
operational noncontextuality.  In particular, it is not a convex mixture of
deterministic two-valued states.

The classical noncontextual bound on the five cyclic events is
\begin{equation}
        \sum_{i=1}^5 p_i\le 2,
\end{equation}
which is the Klyachko--Can--Binicio\u{g}lu--Shumovsky (KCBS) pentagon
inequality~\cite{Klyachko-2008,Badziag-2011,Bub-2009}.  By contrast, a
faithful quantum orthogonal representation reaches
\begin{equation}
        \sum_{i=1}^5 p_i=\sqrt{5}.
\end{equation}
This is the Lov\'asz number of the pentagon exclusivity
graph~\cite{lovasz-79,cabello-severini-winter-2014}.

The half-weight~\eqref{eq:pentagon-half-weight} gives
\(\sum_i p(a_i)=5/2\), so it is admissible and atom-consistent, yet lies
beyond both the classical bound and the usual KCBS/Lov\'asz quantum value.
It is therefore an explicit admissible, atom-consistent, but post-quantum
point in the gap between the classical and Hilbert-space bodies.  In the
language used here, the Hilbert-space squared-projection body escapes the
classical simplex body only after enough context-wise shadows have been
required to glue together; the exotic half-weight escapes still farther
inside the larger admissible/no-disturbance polytope.  It is best read as a
boundary object of the admissible theory, not as a Hilbert-space-realizable
quantum state.

The KCBS inequality can be written as an exact Farkas certificate, but the
all-ones row defining the inequality is not by itself in the sign convention
of Eq.~\eqref{eq:farkas}.  Let
\[
        c^Tb=\sum_{i=1}^5 p(a_i),
\]
and augment the system \(Mq=b\) by a normalization row \(n^Tb=1\), or,
equivalently, use the corresponding linear combination of the
context-normalization rows.  Every deterministic assignment obeys
\[
        c^TM\le 2n^TM.
\]
Hence the Farkas vector \(y=2n-c\) satisfies
\begin{equation}
        y^TM\ge 0,\qquad
        y^Tb=2-\frac52=-\frac12<0.
        \label{eq:pentagon-farkas-certificate}
\end{equation}
This is precisely the dual certificate for the nonexistence of a nonnegative
convex combination of deterministic valuations realizing
\(b=(1/2,\ldots,1/2)\).

The half-weight can be approached inside the admissible polytope by the
one-parameter path
\begin{equation}
        p_r(a_i)=\frac{1}{2+r},\qquad
        p_r(x_i)=\frac{r}{2+r},\qquad r\ge 0.
        \label{eq:pentagon-admissible-path}
\end{equation}
The uniform point occurs at \(r=1\), whereas the exotic half-weight is the
boundary limit \(r\to0\).

This path also admits a generalized-softmax reading.  If the two cyclic atoms
in each context receive a common score \(u_0\), the context-specific atom
receives a score \(t\), and
\[
        r=\frac{g(t)}{g(u_0)},
\]
then Eq.~\eqref{eq:local-softmax} gives precisely
Eq.~\eqref{eq:pentagon-admissible-path}.  For a strictly positive exponential
link, \(r=0\) is reached only as a limiting response rather than at finite
scores; a Heaviside-type link can instead reach the boundary
discontinuously.  This algebraic representation is only an admissible
interpolation.  It does not claim that every point on the path is generated
by one calibrated physical mechanism.  Rather, it reinforces the distinction
between static consistency, classical embeddability, Hilbert-space
realizability, and a calibrated dynamical response law.

\section{Freeze the context and co-vary the preparation}
One may equally well freeze the context and co-vary the preparation: a
half-wave plate before a fixed analyzer, or a calibrated pulse before a
fixed detector basis.  For a single context the two sweeps are related by
moving the group action from the projectors to the state inside the Born
bracket, and the calibration requirements apply verbatim to the preparation
knob.  An \emph{uncalibrated} state sweep, by contrast, discriminates
nothing: by Proposition~1 it covers the entire simplex.  At the
multi-context level the symmetry is only partial.  The gluing problem of
Sec.~\ref{sec:two-contexts} is generated by atoms shared between
measurement contexts, at one fixed state; the nontrivial state-side
analogue is instead preparation noncontextuality in the sense of Spekkens,
in which distinct convex decompositions of the same mixed state must
receive one ontological representation.  The simplex-embedding
tests~\cite{schmid-2021-simplex,selby-2024-lp} treat both sides
symmetrically.

\section{Conclusion}

The moral is not that quantum probability is secretly classical, nor that a
single probability vector contains enough information to reveal its origin.
The moral is more local: one must keep track of what has been thrown away.
When the generating geometry is forgotten, a single context leaves only a
shadow.  When the physical knob is retained, the shadow begins to move, and
its calibrated, symmetry-constrained motion can distinguish specified affine,
Malus-type, softmax, and threshold response models.

A calibrated apparatus sweep nevertheless distinguishes only model classes
whose physical action, covariance, and causal restrictions have been
specified.  Reparameterizing a plotted scalar is not physically harmless
unless the reparameterization intertwines the independently calibrated group
action.  Conversely, one observed response curve does not by itself exclude
arbitrary classical stochastic models or establish a Bell- or
Kochen--Specker-type impossibility theorem.

Two weakly pasted shadows still need not betray anything nonclassical.  If two
contexts agree on their shared atoms, the remaining mass can be coupled by an
ordinary joint distribution.  The obstruction appears only when enough local
shadows are required to hang together as one global object.  The question
then becomes a linear feasibility problem: either there is a classical joint
model or simplex embedding, or Farkas' lemma supplies a separating
inequality.  The elastic-band models provide a laboratory parable for this
distinction: with adaptive realignment, context communication, or one
extended nonseparable object, they can produce quantum-looking pairwise
shadows, but not one Bell-local valuation glued across all four settings.

Thus the staged picture is: shadow, calibrated knob, weak pasting, global
web.  The frozen shadow hides classical, quantum, and exotic origins.  The
knob records how the shadow changes under symmetry-constrained variation.
Weak pastings remain extendable.  Rich webs can fail to extend.  Finally, the
pentagon half-weight reminds us that admissibility and atom-consistency are
not the same as classical embeddability, Hilbert-space realizability, or a
physical response law: the same static constraints can be obeyed by very
different knobs.

\begin{acknowledgments}
This research was funded in whole or in part by the
\textit{Austrian Science Fund (FWF)}
[Grant
\href{https://doi.org/10.55776/PIN5424624}
{digital object identifier (DOI): 10.55776/PIN5424624}].
The author acknowledges TU Wien Bibliothek for financial support through its
Open Access Funding Programme.

This text was partially created and revised with assistance from large
language models.  All content, ideas, and prompts were provided by the author.
\end{acknowledgments}

\bibliography{svozil}

\end{document}